\PassOptionsToPackage{dvipdfmx}{graphicx}
\documentclass{article}
\usepackage{graphicx}

\usepackage[whole]{bxcjkjatype}

\usepackage{amsthm}
\newtheorem{theorem}{Theorem}

\newtheorem{lemma}{Lemma}
\newtheorem{observation}{Observation}
\theoremstyle{definition}
\newtheorem{example}{Example}

\usepackage[T1]{fontenc}
\usepackage{thmtools}
\usepackage{thm-restate}
\usepackage{ascmac}
\usepackage{here}
\usepackage{lineno}
\usepackage{enumitem}
\usepackage{comment}
\usepackage{algorithm}
\usepackage{algorithmic}
\usepackage[caption=false]{subfig}
\usepackage{color}
\usepackage{setspace}
\usepackage{booktabs}

\newcommand{\nil}{\mathit{nil}}
\newcommand{\dif}{\mathit{dif}}
\newcommand{\suf}{\mathit{suf}}
\newcommand{\link}{\mathsf{link}}
\newcommand{\serieslink}{\mathsf{serieslink}}
\newcommand{\paltree}{\mathsf{paltree}}
\newcommand{\seriestree}{\mathsf{seriestree}}

\newcommand{\LSufPalArray}{\mathsf{LSufPal}}
\newcommand{\LPrePalArray}{\mathsf{LPrePal}}
\newcommand{\waq}{\mathsf{WAQ}}

\newcommand{\weight}{\mathsf{weight}}
\newcommand{\RMQ}{\mathsf{RMQ}}
\newcommand{\MP}{\mathsf{MP}}
\newcommand{\order}{\mathcal{O}}

\newcommand{\TopLPal}{\mathsf{TopLPal}}
\newcommand{\LPP}{\mathit{LPP}}
\newcommand{\LPS}{\mathit{LPS}}

\usepackage{authblk}

\begin{document}

\title{Finding Top-$k$ Longest Palindromes in Substrings}

\author[1]{Kazuki Mitani}
\author[2]{Takuya Mieno}
\author[3]{Kazuhisa Seto}
\author[3]{Takashi Horiyama}

\affil[1]{Graduate School of Information Science and Technology, Hokkaido University, Sapporo, Japan}
\affil[2]{Department of Computer and Network Engineering, University of Electro-Communications, Chofu, Japan}
\affil[3]{Faculty of Information Science and Technology, Hokkaido University, Sapporo, Japan}

\date{}

\maketitle

\begin{abstract}
  Palindromes are strings that read the same forward and backward.
  Problems of computing palindromic structures in strings have been studied for many years with a motivation of their application to biology.
  The longest palindrome problem is one of the most important and classical problems regarding palindromic structures,
  that is, to compute the longest palindrome appearing in a string $T$ of length $n$.
  The problem can be solved in $\order(n)$ time by the famous algorithm of Manacher [Journal of the ACM, 1975].
  This paper generalizes the longest palindrome problem to the problem of finding top-$k$ longest palindromes in an arbitrary substring, including the input string $T$ itself.
  The internal top-$k$ longest palindrome query is, given a substring $T[i..j]$ of $T$ and a positive integer $k$ as a query, to compute the top-$k$ longest palindromes appearing in $T[i.. j]$.
  This paper proposes a linear-size data structure that can answer internal top-$k$ longest palindromes query in optimal $O(k)$ time.
  Also, given the input string $T$, our data structure can be constructed in $\order(n\log n)$ time.
  For $k = 1$, the construction time is reduced to $\order(n)$.
\end{abstract}

\section{Introduction}
Palindromes are strings that read the same backward as forward.
Palindromes have been widely studied with the motivation of their application to biology~\cite{gusfield1997algorithms}.
Computing and counting palindromes in a string are fundamental tasks.
Manacher~\cite{manacher1975new} proposed an $\order(n)$-time algorithm that computes all maximal palindromes in the string of length $n$.
Droubay et al.~\cite{droubay2001episturmian} showed that any string of length $n$ contains at most $n+1$ distinct palindromes (including the empty string).
Then, Groult et al.~\cite{groult2010counting} proposed an $\order(n)$-time algorithm to enumerate the number of distinct palindromes in a string.
The above $\order(n)$-time algorithms are time-optimal
since reading the input string of length $n$ takes $\Omega(n)$ time.

Regarding the longest palindrome computation,
Funakoshi et al.~\cite{funakoshi2021computing} considered the problem of computing the longest palindromic substring of the string $T'$ after a single character insertion, deletion, or substitution is applied to the input string $T$ of length $n$.
Of course, using $\order(n)$ time, we can obtain the longest palindromic substring of $T'$ from scratch.
However, this idea is na\"ive and appears to be inefficient.
To avoid such inefficiency, Funakoshi et al.~\cite{funakoshi2021computing} proposed an $\order(n)$-space data structure
that can compute the solution for any editing operation given as a query
in $\order(\log( \min \{\sigma,\log n\}))$ time
where $\sigma$ is the alphabet size.
Amir et al.~\cite{amir2020dynamic} considered the dynamic longest palindromic substring problem, which is an extension of Funakoshi et al.'s  problem where up to $\order(n)$ sequential editing operations are allowed.
They proposed an algorithm that solves this problem in $\order(\sqrt{n} \log^2 n)$ time per a single character edit w.h.p. with a data structure of size $\order(n \log n)$, which can be constructed in $\order(n \log^2 n)$ time.
Furthermore, Amir and Boneh~\cite{amir2019dynamic} proposed an algorithm running in poly-logarithmic time per a single character substitution.

Internal queries are queries about substrings of the input string $T$.
Let us consider a situation where we solve a certain problem for each of $k$ different substrings of $T$. 
If we run an $\order(|w|)$-time algorithm from scratch for each substring $w$, the total time complexity can be as large as $\order(kn)$.
To be more efficient, by performing an appropriate preprocessing on $T$, we construct some data structure for the query to output each solution efficiently.
Such efficient data structures for palindromic problems are known.
Rubinchik and Shur~\cite{rubinchik2017counting} proposed an algorithm that
computes the number of distinct palindromes in a given substring of an input string of length $n$.
Their algorithm runs in $\order(\log n)$ time
with a data structure of size $\order(n \log n)$, which can be constructed in $\order(n \log n)$ time.
Amir et al.~\cite{amir2020dynamic} considered a problem of 
computing the longest palindromic substring in a given substring of the input string of length $n$; it is called the internal longest palindrome query.
Their algorithm runs in $\order(\log n)$ time
with a data structure of size $\order(n \log n)$, which can be constructed in $\order(n \log^2 n)$ time.

This paper proposes a new algorithm for the internal longest palindrome query. The algorithm of Amir et al.~\cite{amir2020dynamic} uses 2-dimensional orthogonal range maximum queries~\cite{agarwal2017range,alstrup2000,bentley1980multidimensional};
furthermore, the time and space complexities of their algorithm are dominated by this query.
Instead of 2-dimensional orthogonal range maximum queries, by using palindromic trees~\cite{rubinchik2018eertree}, weighted ancestor queries~\cite{ganardi2021compression}, and range maximum queries~\cite{fischer2011space}, we obtain a time-optimal algorithm. 

\begin{restatable}{theorem}{maintheorem}
\label{theorem:conclusion}
  Given a string $T$ of length $n$ over a linearly sortable alphabet,
  we can construct a data structure of size $\order(n)$ in $\order(n)$ time that can answer any 
  internal longest palindrome query in $\order(1)$ time.
\end{restatable}
\noindent
Here, an alphabet is said to be \emph{linearly sortable}
if any sequence of $n$ characters from $\Sigma$ can be sorted in $\order(n)$ time.
For example, the integer alphabet $\{1, 2, \ldots, n^c\}$ for some constant $c$ is linearly sortable
because we can sort a sequence from the alphabet in linear time by using a radix sort with base $n$.
We also assume the word-RAM model with word size $\omega \ge \log n$ bits for input size $n$.

Furthermore, we consider a more general problem of finding palindromes,
i.e., the problem of finding \emph{top-$k$} longest palindromes in a substring of $T$,
rather than just the longest palindrome in $T$.
Then, we finally prove the following proposition,
which will be given as a corollary in Section~\ref{sec:topk}.
\begin{restatable}{corollary}{maincoro}
\label{cor:main}
  Given a string $T$ of length $n$,
  we can construct a data structure of size $\order(n)$ in $\order(n\log n)$ time
  that can answer any 
  internal top-$k$ longest palindrome query in $\order(k)$ time.
\end{restatable}
Our results are summarized in Table~\ref{tab:results}.
\begin{table}[tbh]
  \begin{center}
  \begin{tabular}{r|c|c|c}
  \toprule
  &  &  longest palindrome &  top-$k$ palindromes \\
  \hline
  \hline
  String $T$ &  preprocessing & $\order(n)$ time~\cite{manacher1975new} & $\order(n)$ time\\
  \cmidrule{2-4}
  & query & $\order(1)$ time & $\order(k)$ time \\
  \hline
  Query substring &  preprocessing & $\order(n)$ time & $\order(n\log n)$ time\\
  \cmidrule{2-4}
  $T[i.. j]$ & query & $\order(1)$ time & $\order(k)$ time\\
  \bottomrule
  \end{tabular}
  \label{tab:results}
  \caption{
    Manacher's algorithm~\cite{manacher1975new} can compute the longest palindrome in a string $T$ in linear time.
    We study three generalized problems and give efficient data structures and algorithms.
    All the proposed data structures are of linear size.
  }
  \end{center}
\end{table}

\paragraph*{\bf Related Work}
Internal queries have been studied on many problems, not only those related to palindromic structures.
For instance, Kociumaka et al.~\cite{kociumaka2014internal} considered the internal pattern matching queries that are ones for computing the occurrences of a substring $U$ of the input string $T$ in another substring $V$ of $T$.
Besides, internal queries for string alignment~\cite{charalampopoulos2021almost,sakai2019substring,sakai2022data,tiskin2008semi}, longest common prefix \cite{abedin2020linear,amir2014range,ganguly2018linear,matsuda2020compressed}, and longest common substring~\cite{amir2020dynamic} have been studied in the last two decades.
See~\cite{kociumaka2019efficient} for an overview of internal queries.
We also refer to \cite{rangeSUS,babenko2016computing,badkobeh2022internal,charalampopoulos2021internal,charalampopoulos2020faster,kociumaka2016minimal} and references therein.

\paragraph*{\bf Paper Organization}
The rest of this paper is organized as follows.
Section~\ref{sec:pre} gives some notations and definitions.
Section~\ref{sec:algo} shows our data structure to solve the internal longest palindrome queries.
Section~\ref{sec:topk} shows how to compute the top-$k$ longest palindromes in (sub)strings.
Finally, Section~\ref{sec:conc} concludes this paper.

\section{Preliminaries}\label{sec:pre}
\subsection{Strings and Palindromes}
Let $\Sigma$ be an alphabet.
An element of $\Sigma$ is called a character,
and an element of $\Sigma^*$ is called a string.
The empty string $\varepsilon$ is the string of length $0$.
The length of a string $T$ is denoted by $|T|$.
For each $i$ with $1 \le i \le |T|$,
the $i$-th character of $T$ is denoted by $T[i]$.
For each $i$ and $j$ with $1 \le i, j \le |T|$,
the string $T[i]T[i+1] \cdots T[j]$ is denoted by $T[i..j]$.
For convenience, let $T[i'.. j'] = \varepsilon$ if $i' > j'$.
If $T=xyz$, then $x$, $y$, and $z$ are called a prefix, substring, and suffix of $T$, respectively.
They are called a proper prefix, a proper substring, and a proper suffix of $T$ if $x \neq T$, $y \neq T$, and $z \neq T$, respectively.
The string $y$ is called an infix of $T$ if $x \neq \varepsilon$ and $z \neq \varepsilon$.
The reversal of string $T$ is denoted by $T^R$, i.e., $T^R=T[|T|] \cdots T[2]T[1]$.
A string $T$ is called a palindrome if $T=T^R$. Note that $\varepsilon$ is also a palindrome.
For a palindromic substring $T[i.. j]$ of $T$,
the center of $T[i.. j]$ is $\frac{i+j}{2}$.
A palindromic substring $T[i..j]$ is called a maximal palindrome in $T$
if $i=1$, $j=|T|$, or $T[i-1] \neq T[j+1]$.
In what follows, we consider an arbitrary fixed string $T$ of length $n > 0$.
In this paper, we assume that the alphabet $\Sigma$ is linearly sortable.
We also assume the word-RAM model with word size $\omega \ge \log n$ bits.

Let $z$ be the number of palindromic suffixes of $T$. 
Let $\suf(T) = (s_1, s_2, \ldots,$
$s_{z})$ be the sequence of the lengths of palindromic suffixes of $T$ sorted in increasing order.
Further let $\dif_i = s_i - s_{i-1}$ for each $i$ with $2 \le i \le z$.
For convenience, let $\dif_1=0$.
Then, the sequence $(\dif_1, \ldots, \dif_z)$ is monotonically non-decreasing~(Lemma 7 in \cite{matsubara2009efficient}).
Let $(\suf_1, \suf_2, \ldots ,\suf_p)$ be the partition of $\suf(T)$
such that for any two elements $s_i, s_j$ in $\suf(T)$, 
$s_i, s_j \in \suf_k$ for some $k$ iff $\dif_i = \dif_j$.
By definition, each $\suf_k$ forms an arithmetic progression.
It is known that the number $p$ of arithmetic progressions satisfies $p \in \order(\log n)$~\cite{APOSTOLICO1995163,matsubara2009efficient}.
For $1 \leq k \leq p$ and $1 \leq \ell \leq |\suf_k|$, $\suf_{k,\ell}$ denote the $\ell$-th term of $\suf_k$.
Figure~\ref{fig:progression} shows an example of the above definitions. 
\begin{figure}[tbp]
    \centering
    \includegraphics[width=\linewidth]{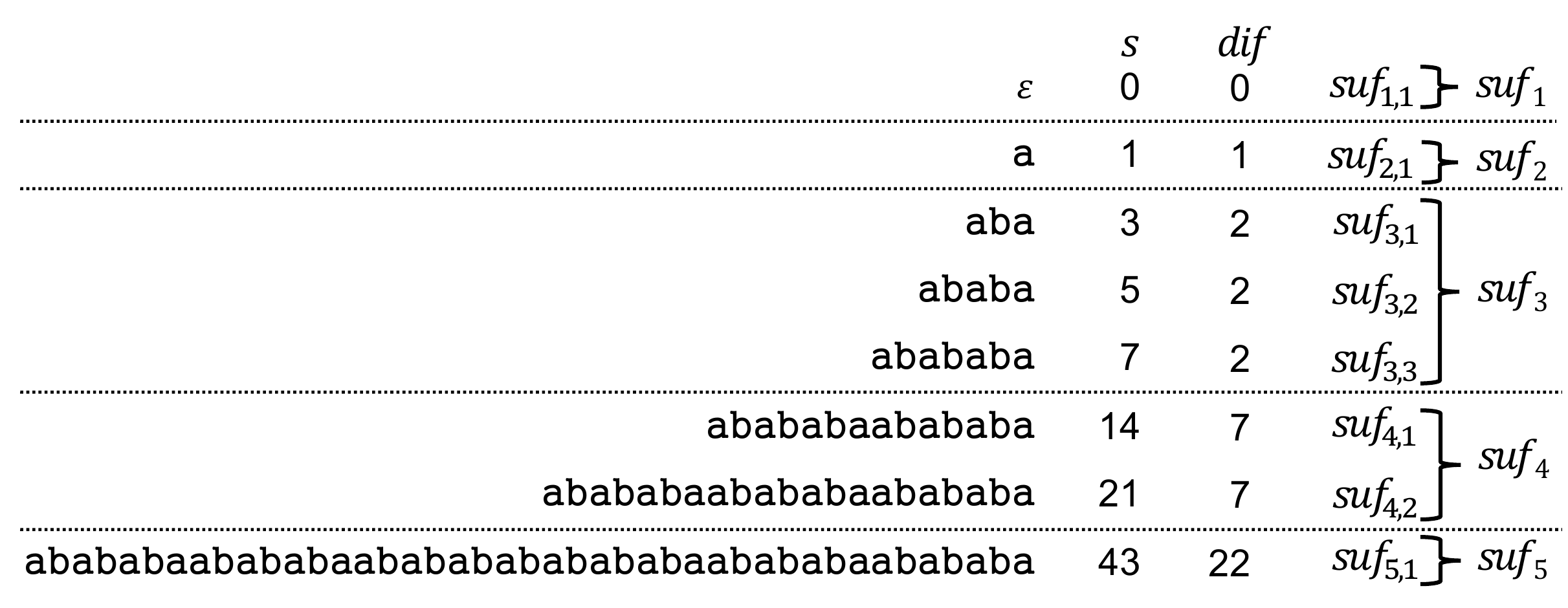}
    \caption{
    Palindromic suffixes of a string
    and the partition $(\suf_1, \ldots, \suf_5)$ of their lengths.
    Three integers $s_3 = 3, s_4 = 5$, and $s_5 = 7$ are represented by a single arithmetic progression $\suf_3$ 
    since $\dif_3=\dif_4=\dif_5=2$.
    Since $s_4$ is the second smallest term in $\suf_3$, $\suf_{3,2} = s_4$.
    }
    \label{fig:progression}
\end{figure}

\subsection{Tools}
In this section, we list some data structures used in our algorithm in Section~\ref{sec:algo}.
\paragraph*{Palindromic Trees and Series Trees}
The palindromic tree of $T$ is a data structure that represents all distinct palindromes in $T$~\cite{rubinchik2018eertree}.
The palindromic tree of $T$, denoted by $\paltree(T)$, has $d$ ordinary nodes and one auxiliary node $\bot$ where $d \le n+1$ is the number of all distinct palindromes in $T$.
Each ordinary node $v$ corresponds to a palindromic substring of $T$ (including the empty string $\varepsilon$) and stores its length as $\weight(v)$.
For the auxiliary node $\bot$, we define $\weight(\bot)=-1$.
For convenience, we identify each node with its corresponding palindrome.
For an ordinary node $v$ in $\paltree(T)$ and a character $c$,
if nodes $v$ and $cvc$ exist, then an edge labeled $c$ connects these nodes.
The auxiliary node $\bot$ has edges to all nodes corresponding to length-$1$ palindromes.
Each node $v$ in $\paltree(T)$ has a suffix link that points to the longest palindromic proper suffix of $v$.
Let $\link(v)$ be the string pointed to by the suffix link of $v$.
We define
$\link(\varepsilon) = \link(\bot) = \bot$.
See Figure~\ref{fig:eertree}(a) for example.
For each node $v$ corresponding to a non-empty palindrome in $\paltree(T)$, let $\delta_v=|v|-|\link(v)|$ be the difference between the lengths of $v$ and its longest palindromic proper suffix.
For convenience, let $\delta_{\varepsilon} = 0$.
Each node $v$ corresponding to a non-empty palindrome has a series link that points to the longest palindromic proper suffix $u$ of $v$ such that $\delta_u \neq \delta_v$.
Let $\serieslink(v)$ be the string pointed to by the series link of $v$.

Let $\LSufPalArray$ be an array of length $n$
such that $\LSufPalArray[j]$ stores a pointer to the node in $\paltree(T)$ corresponding to the longest palindromic suffix of $T[1..j]$ for each $1 \le j \le n$.
The definition of $\LSufPalArray$ is identical to the array $\mathsf{node}[1]$ defined in~\cite{rubinchik2018eertree},
and it was shown that $\mathsf{node}[1]$ can be computed in $\order(n)$ time.
Hence, $\LSufPalArray$ can be computed in $\order(n)$ time.
Let $\LPrePalArray$ be an array of length $n$
such that $\LPrePalArray[i]$ stores a pointer to the node in $\paltree(T)$ corresponding to the longest palindromic prefix of $T[i..n]$ for each $1 \le i \le n$.
$\LPrePalArray$ can be computed in $\order(n)$ time as well as $\LSufPalArray$.
\begin{theorem}[Proposition 4.10 in~\cite{rubinchik2018eertree}]
\label{theorem:paltree}
  Given a string $T$ over a linearly sortable alphabet,
  the palindromic tree of $T$, including its suffix links and series links,
  can be constructed in $\order(n)$ time.
  Also, $\LSufPalArray$ and $\LPrePalArray$ can be computed in $\order(n)$ time.
\end{theorem}

Let us consider the subgraph $\mathcal{S}$ of $\paltree(T)$
that consists of all ordinary nodes and reversals of all series links.
By definition, $\mathcal{S}$ has no cycle, and $\mathcal{S}$ is connected
(any node is reachable from the node $\varepsilon$), i.e., it forms a tree.
We call the tree $\mathcal{S}$ the series tree of $T$ and denote it by $\seriestree(T)$.
By definition of series links, the set of lengths of palindromic suffixes of $v$ that are longer than $|\serieslink(v)|$ can be represented by an arithmetic progression.
Each node $v$ stores the arithmetic progression, represented by a triple consisting of its first term, its common difference, and the number of terms.
Arithmetic progressions for all nodes can be computed in linear time by traversing the palindromic tree.
It is known that the length of a path consisting of series links is $\order(\log n)$~\cite{rubinchik2018eertree}.
Hence, the height of $\seriestree(T)$ is $\order(\log n)$.
See Figure~\ref{fig:eertree}(b) for illustration.
\begin{figure}[tbp]
  \begin{minipage}[t]{0.45\linewidth}
    \centering
    \subfloat[$\paltree(T)$]{\includegraphics[width=5.0cm]{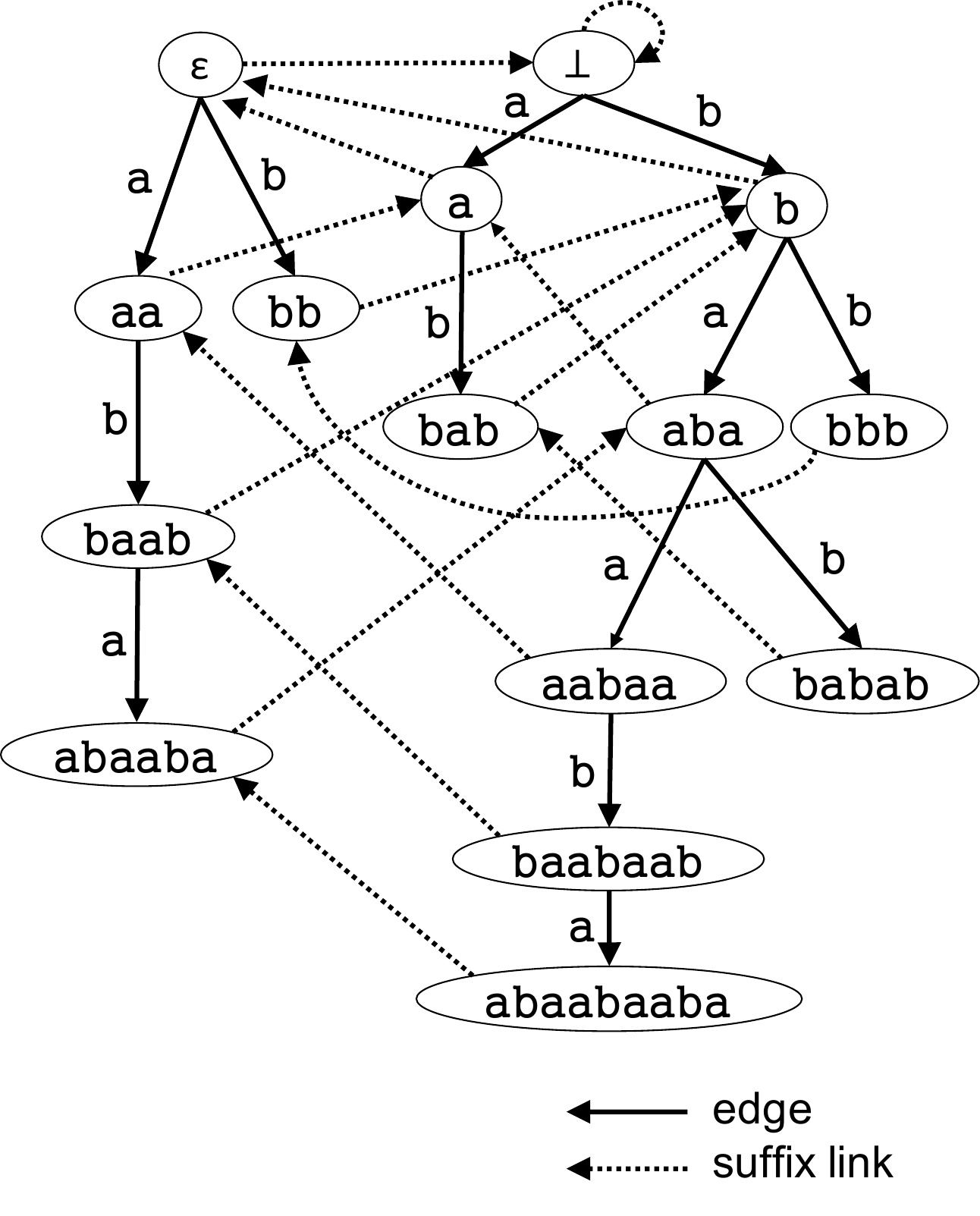}
    \label{fig:eertree1}}
  \end{minipage}
  \begin{minipage}[t]{0.45\linewidth}
    \centering
    \subfloat[$\seriestree(T)$]{\includegraphics[width=5.0cm]{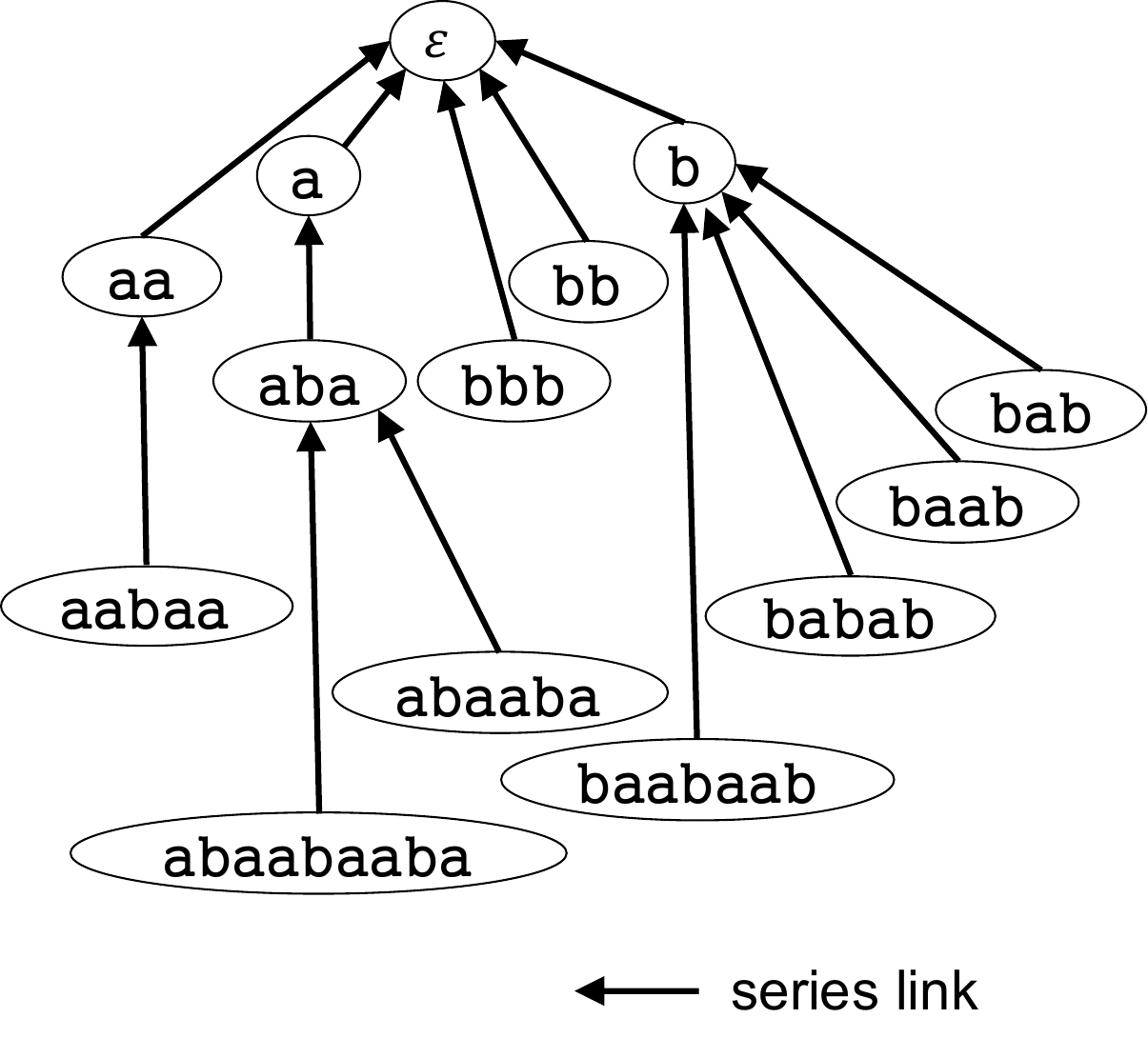}
    \label{fig:eertree2}}
  \end{minipage}
  \caption{
    Illustration for the palindromic tree and the series tree of string $T=\mathtt{abaabaabababbb}$
    Since $\delta_{\mathtt{abaabaaba}}=|\mathtt{abaabaaba}|-|\mathtt{abaaba}|=3$, 
    $\delta_{\mathtt{abaaba}}=|\mathtt{abaaba}|-|\mathtt{aba}|=3$, 
    and $\delta_{\mathtt{aba}}=|\mathtt{aba}|-|\mathtt{a}|=2$,
    then $\serieslink(\mathtt{abaabaaba})=\mathtt{aba}$.
    $\mathtt{abaabaaba}$ stores the arithmetic progression representing $\{6,9\}$,
    $\mathtt{abaaba}$ stores the arithmetic progression representing $\{6\}$
    and $\mathtt{aba}$ stores the arithmetic progression representing $\{3\}$.
  }
  \label{fig:eertree}
\end{figure}

\paragraph{Weighted Ancestor Query}
A rooted tree whose nodes are associated with integer weights is called a monotone-weighted tree
if the weight of every non-root node is not smaller than the parent's weight.
Given a monotone-weighted tree $\mathcal{T}$ for preprocess and a node $v$ and an integer $k$ for query, a weighted ancestor query (WAQ) returns the ancestor $u$ closest to the root of $v$ such that the weight of $u$ is greater than $k$.
Let $\waq_{\mathcal{T}}(v,k)$ be the output of the weighted ancestor query for tree $\mathcal{T}$, node $v$, and integer $k$.
See Figure~\ref{fig:WAQ} for a concrete example.
\begin{figure}[tbp]
    \centering
    \includegraphics[width=0.5\linewidth]{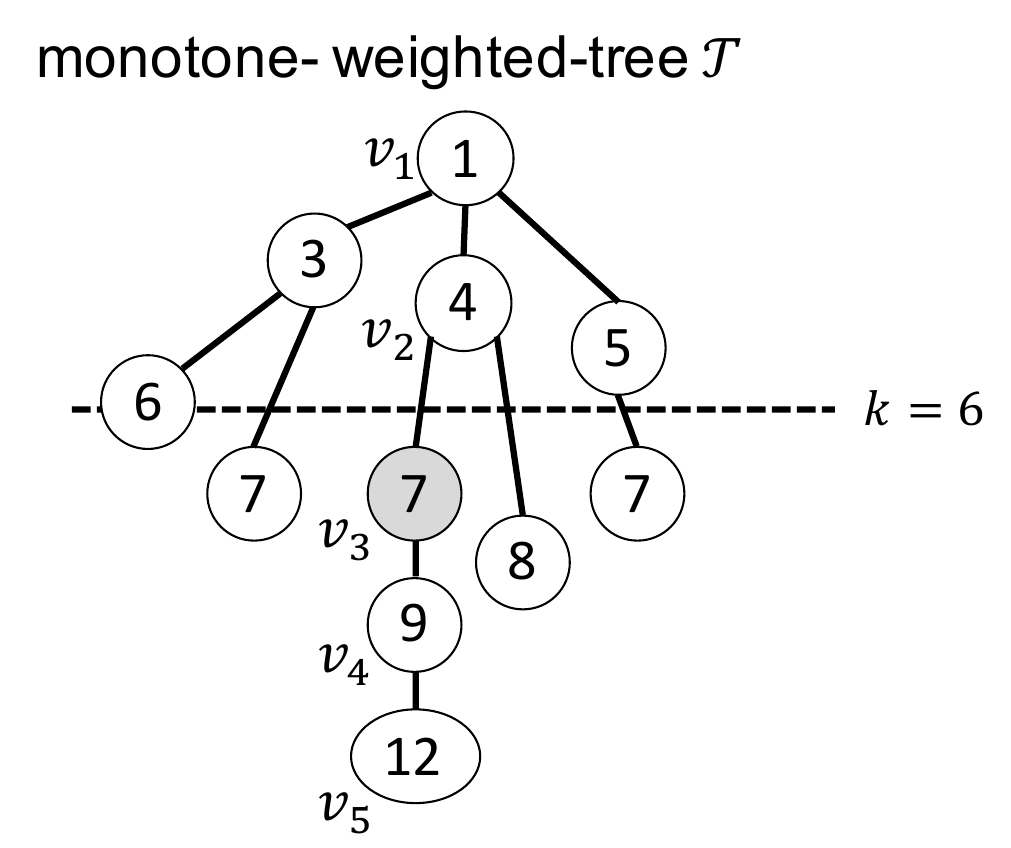}
    \caption{
    Illustration for weighted ancestor query.
    Integers in nodes denote the weights.
    Given a node $v_5$ in a monotone-weighted tree $\mathcal{T}$ and an integer $k=6$ for query,
    WAQ returns the node $v_3$ 
    since $v_3$ is an ancestor of $v_5$, $\weight(v_3)>k=6$, and the weight of the parent $v_2$ of $v_3$ is not greater than $k=6$.
    } \label{fig:WAQ}
\end{figure}
It is known that there is an $\order(N)$-space data structure that can answer any weighted ancestor query in $\order(\log\log N)$ time where $N$ is the number of nodes in the tree~\cite{amir2007pattern}.
In general, the query time $\order(\log\log N)$ is known to be optimal within $\order(N)$ space~\cite{puatracscu2006time}.
On the other hand, if the height of the input tree is low enough, the query time can be improved:
\begin{theorem}[Proposition 15 in~\cite{ganardi2021compression}]
\label{theorem:WAQ}
  Let $\omega$ be the word size.
  Given a monotone-weighted tree with $N$ nodes and height $\order(\omega)$,
  one can construct an $\order(N)$ space data structure in $\order(N)$ time
  that can answer any weighted ancestor query in constant time.
\end{theorem}

In this paper, we use weighted ancestor queries only on the series tree of $T$ whose height is $\order(\log n) \subseteq \order(\omega)$, where $\omega$ is the word size, thus we will apply Theorem~\ref{theorem:WAQ}.
Note that we assume the word-RAM model with word size $\omega \ge \log n$ bits.

\paragraph{Range Maximum Query}
Given an integer array $A$ of length $m$ for preprocess and two indices $i$ and $j$ with $1\leq i\leq j \leq m$ for query,
range maximum query returns the index of a maximum element in the sub-array $A[i..j]$.
Let $\mathrm{RMQ}_A(i, j)$ be the output of the range maximum query for array $A$ and indices $i, j$.
In other words, $\mathrm{RMQ}_A(i,j)=\arg\max_k\{A[k]\mid i \leq k \leq j\}$.
The following result is known:
\begin{theorem}[Theorem 5.8 in~\cite{fischer2011space}]
\label{theorem:RMQ}
  Let $m$ be the size of the input array $A$.
  There is a data structure of size $2m+o(m)$ bits
  that can answer any range maximum query on $A$ in constant time.
  The data structure can be constructed in $\order(m)$ time.  
\end{theorem}

\section{Internal Longest Palindrome Queries}\label{sec:algo}
In this section, we propose an efficient data structure for the internal longest palindrome query defined as follows:
\begin{itembox}[l]{\bf Internal longest palindrome query}
  {\bf Preprocess:} A string $T$ of length $n$.\\
  {\bf Query input:} Two indices $i$ and $j$ with $1 \leq i \leq j \leq n$.\\
  {\bf Query output:} The longest palindromic substring in $T[i..j]$.
\end{itembox}
Our data structure requires only $\order(n)$ words of space
and can answer any internal longest palindrome query in constant time.
To answer queries efficiently, we classify all palindromic substrings of $T$ into
palindromic prefixes, palindromic infixes, and palindromic suffixes.
First, we compute the longest palindromic prefix and the longest palindromic suffix of $T[i..j]$.
Second, we compute a palindromic infix that is a candidate for the answer.
As we will discuss in a later subsection, this candidate may not be the longest palindromic infix of $T[i.. j]$.
Finally, we compare the three above palindromes and output the longest one.

\subsection{Palindromic Suffixes and Prefixes}
First, we compute the longest palindromic suffix of $T[i..j]$.
In the preprocessing, we build $\seriestree(T)$ and a data structure for the weighted ancestor queries on $\seriestree(T)$, and compute $\LSufPalArray$ as well.
The query algorithm consists of three steps:
\begin{description}
  \item[Step 1: Obtain the longest palindromic suffix of ${T[1..j]}$.]\mbox{}\\
    We obtain the longest palindromic suffix $v$ of $T[1..j]$ from $\LSufPalArray[j]$.
    If $|v| \leq |T[i..j]|$, then $v$ is the longest palindromic suffix of $T[i.. j]$.
    Then we return $T[j-|v|+1.. j]$, and the algorithm terminates.
    Otherwise, we continue to Step 2.
  \item[Step 2: Determine the group to which the desired length belongs.]\mbox{}\\
    Let $\ell$ be the length of the longest palindromic suffix of $T[i..j]$ we want to know.
    We use the longest palindromic suffix $v$ of $T[1..j]$ obtained in Step 1.
    First, we find the shortest palindrome $u$ that is an ancestor of $v$ in $\seriestree(T)$ and has a length at least $|T[i..j]|$.
    Such a palindrome (equivalently the node) $u$ can be found by weighted ancestor query on the series tree,
    i.e., $u = \waq_{\seriestree(T)}(v, j-i)$. 
    Then $|u|$ is an upper bound of $\ell$.
    Let $\suf_\alpha$ be the group such that $|u| \in \suf_\alpha$.
    If the smallest element $\suf_{\alpha, 1}$ in $\suf_\alpha$ is at most $|T[i..j]|$, 
    the length $\ell$ belongs to the same group $\suf_{\alpha}$ as $|u|$.
    Otherwise, the length $\ell$ belongs to the previous group $\suf_{\alpha-1}$.
  \item[Step 3: Calculate the desired length.]\mbox{}\\
    Let $\suf_\beta$ be the group to which the length $\ell$ belongs, which is determined in Step 2.
    Since $\suf_{\beta}$ is an arithmetic progression,
    i.e., $\suf_{\beta,\gamma}=\suf_{\beta,1}+(\gamma-1)\dif_\beta$ for $1 \leq \gamma \leq |\suf_\beta|$,
    the desired length $\ell$ can be computed by using a constant number of arithmetic operations.
    Then we return $T[j-\ell+1..j]$.
\end{description}
\vspace{\baselineskip}
See Figure~\ref{fig:suffix} for illustration.
Now, we show the correctness of the algorithm and analyze time and space complexities.

\begin{figure}[t]
    \centering
    \includegraphics[width=\linewidth]{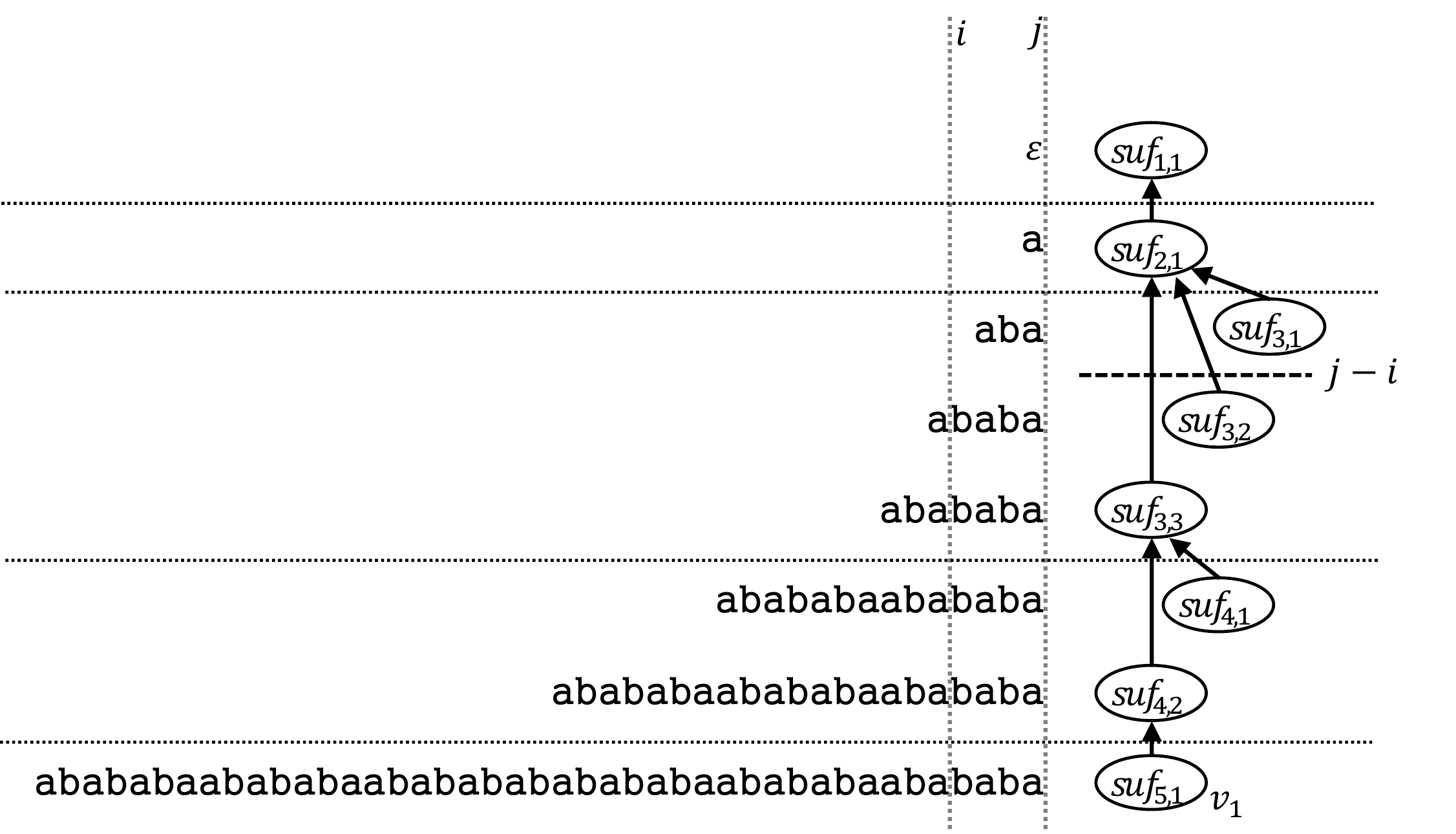}
    \caption{
    Illustration for how to compute the longest palindromic suffix of $T[i..j]$, when $T[1..j]=\mathtt{abababaabababaabababababababaabababaabababa}$.
    The graph on the right hand depicts a part of the series tree of a string $T$, and the lengths of palindromes are written inside the nodes.
    In Step 1, we obtain the length $\suf_{5,1}$ of the longest palindromic suffix $v_1$ of $T[1.. j]$.
    In Step 2, we find $\suf_{3,3}$ by $\waq_{\seriestree(T)}(v_1, j-i)$.
    Since $\suf_{3,1} > j-i+1$, the desired length belongs to $\suf_3$.
    In Step 3, since $\suf_3$ is an arithmetic progression, we can find that $\suf_{3,1}$ is the longest palindromic suffix of $T[i..j]$ in constant time.
    }
    \label{fig:suffix}
\end{figure}

\begin{lemma}
\label{lemma:suffix}
  We can compute the longest palindromic suffix and prefix of $T[i..j]$ in $\order(1)$ time
  with a data structure of size $\order(n)$ that can be constructed in $\order(n)$ time.
\end{lemma}
\begin{proof}
In the preprocessing, we build $\seriestree(T)$, $\LSufPalArray$, $\LPrePalArray$ and a data structure of weighted ancestor query on $\seriestree(T)$
in $\order(n)$ time~(Theorem~\ref{theorem:paltree} and \ref{theorem:WAQ}).
Recall that since the height of $\seriestree(T)$ is $\order(\log n) \subseteq \order(\omega)$, we can apply Theorem~\ref{theorem:WAQ} to the series tree.
Again, by Theorem~\ref{theorem:paltree} and \ref{theorem:WAQ},
the space complexity is $\order(n)$ words of space.

In what follows, let $\ell$ be the length of the longest palindromic suffix of $T[i..j]$.
In Step 1, we can obtain the longest palindromic suffix $v$ of $T[1..j]$ by just referring to $\LSufPalArray[j]$.
If $|v| \leq |T[i..j]|$, $v$ is also the longest palindromic suffix of $T[i..j]$, i.e., $\ell = |v|$.
Otherwise, $v$ is not a substring of $T[i..j]$.
In Step 2, we first query $\waq_{\seriestree(T)}(v, j-i)$.
The resulting node $u$ corresponds to a palindromic suffix of $T[1.. j]$, which is longer than $|T[i.. j]|$.
Let $\suf_\alpha$ and $\suf_\beta$ be the groups to which $|u|$ and $\ell$ belong to, respectively.
If the smallest element $\suf_{\alpha, 1}$ in $\suf_\alpha$ is at most $j-i+1$,
then the desired length $\ell$ satisfies $\suf_{\alpha,1} \le \ell \le |u|$.
Namely, $\beta = \alpha$.
Otherwise, if $s$ is greater than $j-i+1$, $\ell$ is not in $\suf_\alpha$
but is in $\suf_{\alpha-x}$ for some $x > 1$.
If we assume that $\ell$ belong to $\suf_{\alpha-y}$ for some $y \ge 2$, 
the length of $\serieslink(u)$ belonging to $\suf_{\alpha-1}$ is longer than $T[i.. j]$.
However, it contradicts that $u$ is the answer of $\waq_{\seriestree(T)}(v, j-i)$.
Hence, if $s$ is greater than $j-i+1$, then the length $\ell$ is in $\suf_{\alpha-1}$.
Namely, $\beta = \alpha-1$.
In Step 3,
we can compute $\ell$ in constant time
since we know the arithmetic progression $\suf_\beta$ to which $\ell$ belongs.
More specifically, $\ell$ is the largest element that is in $\suf_\beta$ and is at most $j-i+1$.

Throughout the query algorithm,
all operations, including $\waq$ and operations on arithmetic progressions, can be done in constant time.
Thus the query algorithm runs in constant time.
\end{proof}

We can compute the longest palindromic prefix of $T[i..j]$ in a symmetric way using $\LPrePalArray$ instead of $\LSufPalArray$.

\subsection{Palindromic Infixes}
Next, we compute the longest palindromic infix except for ones that are obviously shorter than
the longest palindromic prefix or the longest palindromic suffix of the query substring.
We show that to find the desired palindromic infix, it suffices to consider maximal palindromes whose centers are between the centers of the longest palindromic prefix and the longest palindromic suffix of $T[i.. j]$.
Let $t$ be the ending position of the longest palindromic prefix and
$s$ be the starting position of the longest palindromic suffix.
Namely, $T[i..t]$ is the longest palindromic prefix
and $T[s..j]$ is the longest palindromic suffix of $T[i..j]$.

\begin{lemma}
\label{lemma:infix1}
    Let $w$ be a palindromic infix of $T[i..j]$ and $c$ be the center of $w$.
    If $c <\frac{i+t}{2}$ or $c > \frac{s+j}{2}$, $w$ cannot be the longest palindromic substring of $T[i..j]$.
\end{lemma}
\begin{proof}
Palindrome $w$ is a proper substring of $T[i..t]$ (resp. $T[s..j]$) if $c <\frac{i+t}{2}$ (resp. $c > \frac{s+j}{2}$).
Then, $w$ is shorter than $T[i..t]$ or $T[s..j]$~(see also Figure~\ref{fig:infix1}).
\end{proof}

\begin{figure}[tbp]
    \centering
    \includegraphics[width=0.5\linewidth]{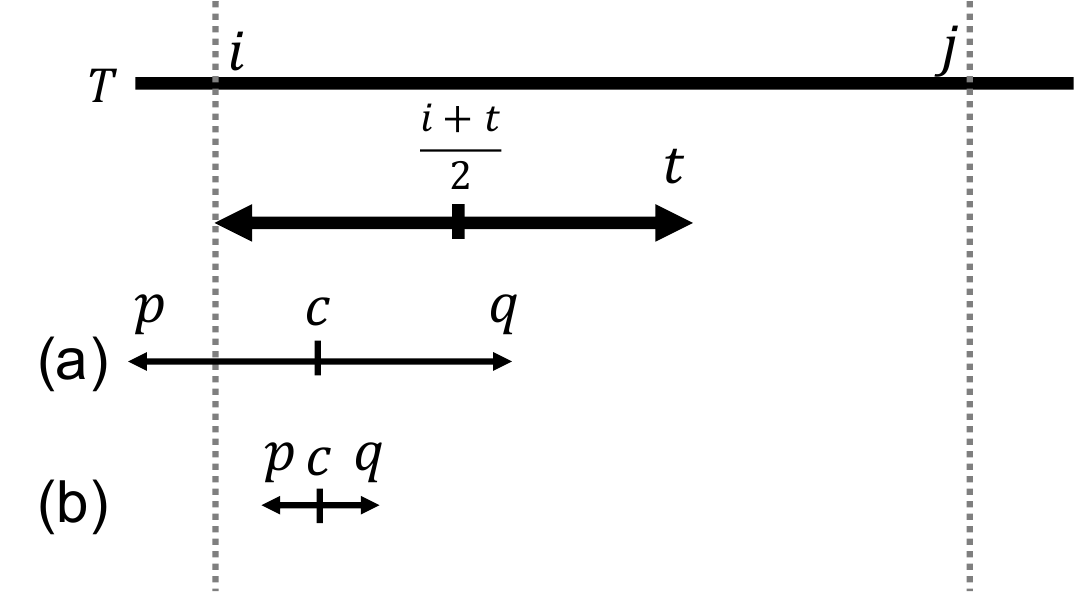}
    \caption{
    Illustration for Lemma~\ref{lemma:infix1}.
    Two-way arrows denote palindromic substrings of $T$.
    $T[i..t]$ is the longest palindromic prefix of $T[i..j]$.
    A palindrome whose center $c$ is less than $\frac{i+t}{2}$ is 
    either (a) not a substring of $T[i..j]$ or (b) shorter than the longest palindromic prefix of $T[i..j]$ as shown in this figure.
    }
    \label{fig:infix1}
     \centering
    \includegraphics[width=0.7\linewidth]{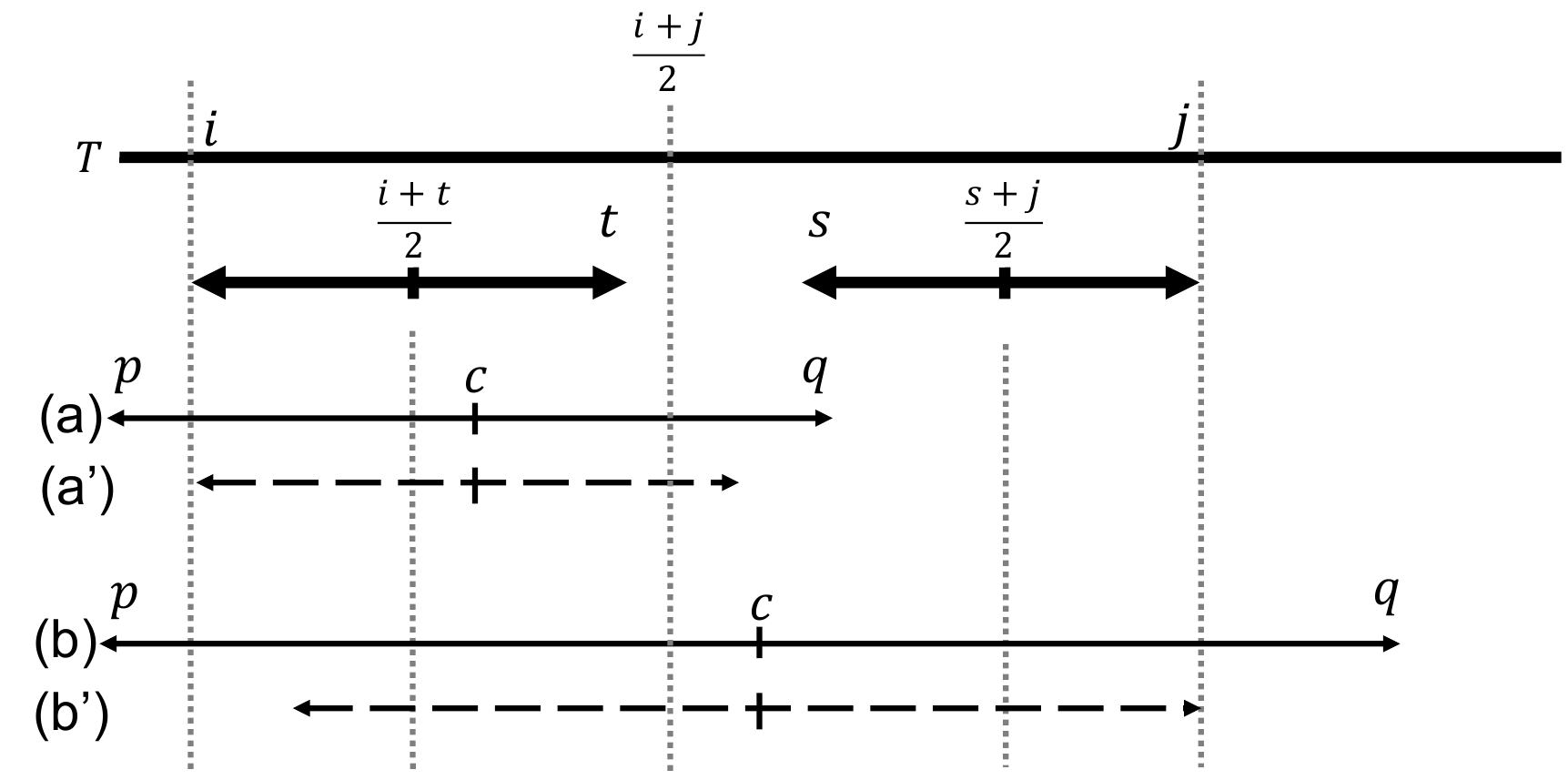}
    \caption{
    Illustration for contradictions in the proof of Lemma~\ref{lemma:infix2}.
    $T[i..t]$ is the longest palindromic prefix and $T[s..j]$ is the longest palindromic suffix of $T[i..j]$.
    If a palindrome as (a) exists, there exists a palindromic prefix (a') of $T[i.. j]$ that is longer than $T[i..t]$, a contradiction.
    Similarly, the existence of a palindrome as (b) leads to a contradiction.
    }
    \label{fig:infix2}
\end{figure}

Then, we consider palindromes whose centers are between the centers of the longest palindromic prefix and the longest palindromic suffix of $T[i.. j]$.

\begin{lemma}
\label{lemma:infix2}
  Let $w$ be a palindromic substring of $T$ and $c$ be the center of $w$.
  If $\frac{i+t}{2} < c < \frac{s+j}{2}$, then $w$ is a palindromic infix of $T[i..j]$.
\end{lemma}
\begin{proof}
Let $w = T[p.. q]$. Then, $c = \frac{p+q}{2}$.
To prove that $w$ is a palindromic infix, we show that $p > i$ and $q < j$.
For the sake of contradiction, we assume $p \leq i$.
If $\frac{i+t}{2}<c \leq \frac{i+j}{2}$, there exists a palindromic prefix $w_1$ whose center is $c$.
This contradicts that $T[i.. t]$ is the longest palindromic prefix of $T[i.. j]$ since $T[i..t]$ is a substring of $w_1$~(see also Figure~\ref{fig:infix2}).
Otherwise, if $\frac{i+j}{2}<c<\frac{s+j}{2}$,
there exists a palindromic suffix whose $w_2$ center is $c$.
This contradicts that $T[s..j]$ is the longest palindromic suffix of $T[i.. j]$ since $T[i..t]$ is a substring of $w_2$~(see also Figure~\ref{fig:infix2}).
Therefore, $p > i$.
We can show $q < j$ in a symmetric way.
\end{proof}

By Lemmas~\ref{lemma:infix1} and~\ref{lemma:infix2},
when a palindromic infix $w$ of $T[i.. j]$ is the longest palindromic substring of $T[i.. j]$,
the center of $w$ must be located between $\frac{i+t}{2}$ and $\frac{s+j}{2}$.
Furthermore, $w$ is a maximal palindrome in $T$.
In other words, $w$ is the longest maximal palindrome in $T$
whose center $c$ satisfies $\frac{i+t}{2} < c < \frac{s+j}{2}$.
To find such a (maximal) palindrome,
we build a succinct RMQ data structure on the length-$(2n-1)$ array $\MP$
that stores the lengths of maximal palindromes in $T$.
For each integer and half-integer $c \in \{1, 1.5, \ldots, n-0.5, n\}$, $\MP[2c-1]$ stores the length of the maximal palindrome whose center is $c$.
By doing so, when the indices $t$ and $s$ are given,
we can find a candidate for the longest palindromic substring
which is an infix of $T[i..j]$ in constant time.
More precisely, the length of the candidate is $\MP[\RMQ_{\MP}(i+t, s+j-2)]$ since the center $c$ of the candidate satisfies $\frac{i+t}{2}<c<\frac{s+j}{2}$ ($i+t-1<2c-1<s+j-1$).
By Manacher's algorithm~\cite{manacher1975new},  $\MP$ can be constructed in $\order(n)$ time.
Then, we obtain the following lemma.

\begin{lemma}
\label{lemma:infix}
  Given the longest palindromic prefix $T[i..t]$ and the longest palindromic suffix $T[s..j]$ of $T[i..j]$,
  we can compute the longest palindromic infix of $T[i..j]$ whose centers are between the centers of $T[i..t]$ and $T[s..j]$
  in $\order(1)$ time with a data structure of size $\order(n)$ that can be constructed in $\order(n)$ time.
\end{lemma}

By Lemmas~\ref{lemma:suffix} and~\ref{lemma:infix}, we have shown our main theorem:
\maintheorem*

\section{Top-$k$ longest palindromes}\label{sec:topk}

We denote by $\TopLPal_T([i, j], k)$ an array of occurrences of top-$k$ longest palindromic substrings in $T[i.. j]$ sorted in their lengths.
In other words,
$\TopLPal_T([i, j], k)[r] = [s, t]$ means that $T[s.. t]$ is the $r$-th longest palindromic substring in $T[i..j]$.
For simplicity, we denote $\TopLPal_T([1, |T|], k)$ as $\TopLPal_T(k)$.

\subsection{Top-$k$ longest palindromes in a string}

First, we consider a problem to compute top-$k$ longest palindromes in the input string and propose an efficient algorithm.
\begin{itembox}[l]{\bf Top-$k$ longest palindromes problem}
  {\bf Input:} A string $T$ of length $n$ and an integer $k$.\\
  {\bf Output:} An array $\TopLPal_T(k)$.
\end{itembox}

Firstly, we give an important observation of this problem.
For a palindromic substring $P = T[\alpha.. \beta]$,
we call the substring $T[\alpha+1.. \beta-1]$ the \emph{shrink} of $P$.
Note that the shrink of a palindrome is also a palindrome.
\begin{observation} \label{obs:topkpal}
  The $r$-th longest palindrome in a string $T$ is either
  \begin{itemize}
    \item[(i)] a maximal palindrome in $T$ or
    \item[(ii)] the shrink of the $q$-th longest palindrome for some $q$ with $1 \le q \le r-1$.
  \end{itemize}
\end{observation}

In our algorithm, we precompute array $M[1.. 2n-1]$
and dynamically maintain array $R[1.. n]$, where
each $M[p]$ stores the $p$-th longest maximal palindromes in $T$, and
each $R[\ell]$ stores the set of palindromes of length $\ell$ that are already returned and whose palindromic substring is not returned yet
(if there is no such palindrome, let $R[\ell]=\nil$).

Note that the sorted array $M$ can be computed in $\order(n)$ time by using Manacher's algorithm and radix sorting.
Also, since the longest palindrome is $M[1]$,
we first return $M[1] = [s, t]$ and then update $R[t-s+1]$ to singleton $\{[s, t]\}$.

When we compute the $r$-th longest palindrome for some $r$ with $2 \le r \le k$,
we utilize Observation~\ref{obs:topkpal}.
Let $\ell_{r-1}$ be the length of the $(r-1)$-th longest palindrome.
Then the length $\ell_r$ of the $r$-th longest palindrome is in $\{\ell_{r-1}, \ell_{r-1}-1, \ell_{r-1}-2\}$
because the shrink of the $(r-1)$-th longest palindrome, whose length is  $\ell_{r-1}-2$, is at least a palindrome.
We pick up a longest palindrome $Q$ within $R[\ell_{r-1}+2]\cup R[\ell_{r-1}+1]\cup R[\ell_{r-1}]$.
Note that such $Q$ always exists since $R[\ell_{r-1}] \neq \nil$ at this step.
Let $Q^-$ be the shrink of $Q$.
We compare its length $|Q^-|$ with the length of the longest maximal palindrome that has not returned yet.
Then the longer one is the $r$-th longest palindrome $P_r$.
If $P_r$ is equal to $Q^-$, we remove $Q$ from $R[|Q|]$.
At last, we append $P_r$ to the set $R[|P_r|]$.

Since every operation in each $r$-th step can be done in constant time,
the above algorithm runs in $\order(n+k)$ time.
Also, since we remove $Q$ from $R$ when we return $Q^-$,
most entries of $R$ are $\nil$ except for at most three entries at each step.
Thus, array $R$ can be implemented within $\min\{3n, k\}$ words of space.
In total, our algorithm requires $\order(n)$ working space.
We have shown the next lemma.
\begin{lemma}\label{lem:topkManacher}
  Given a string $T$ of length $n$ and an integer $k$,
  we can compute $\TopLPal_T(k)$
  in $\order(n+k)$ time with $\order(n)$ working space.
\end{lemma}

We further show that the above algorithm can be applied to a query version of the top-$k$ longest palindromes problem defined as below:
\begin{itembox}[l]{\bf Top-$k$ longest palindromes query}
  {\bf Input:} A string $T$ of length $n$.\\
  {\bf Query input:} An integer $k$.\\
  {\bf Query output:} An array $\TopLPal_T(k)$.
\end{itembox}

In the preprocessing phase, we compute and store the sorted top-$n$ longest palindromes
by using the algorithm of Lemma~\ref{lem:topkManacher}.
If $k \le n$, we just scan the pre-stored palindromes and return the top-$k$ ones.
Otherwise, we apply the aforementioned algorithm.
For both cases, the query time is $\order(k)$, which is optimal.
Thus the following theorem holds.

\begin{theorem}
  After $\order(n)$-time preprocessing on an input string $T$ of length $n$,
  we can compute $\TopLPal_T(k)$ in $\order(k)$ time
  for a given query integer $k$.
\end{theorem}

\subsection{Internal Top-$k$ longest palindromes query}

This subsection considers a more general model; the internal query model.
\begin{itembox}[l]{\bf Internal Top-$k$ longest palindromes query}
  {\bf Input:} A string $T$ of length $n$.\\
  {\bf Query input:} An interval $[i, j]$ and an integer $k$.\\
  {\bf Query output:} An array $\TopLPal_T([i,j], k)$.
\end{itembox}

First, we give some observations and the idea of our algorithm.
The following observation on palindromic symmetry is fundamental.
\begin{observation}\label{observation:mirror}
  Let $T[\alpha..\beta]$ be a palindrome and let $c$ be an (half-) integer with $\alpha \leq c < \frac{\alpha+\beta}{2}$.
  There is a palindromic substring of $T[\alpha..\beta]$ whose center is $c$
  iff
  there is a palindromic substring of $T[\alpha..\beta]$ whose center is $\alpha+\beta-c$.
\end{observation}

Similar to Observation~\ref{obs:topkpal}, we categorize the $r$-th longest palindrome.
Let $\LPP_{i,j}$ (resp. $\LPS_{i,j}$) be the longest palindromic prefix (resp. suffix) of $T[i.. j]$.
Further let $c_p$ (resp. $c_s$) be the center of $\LPP_{i,j}$ (resp. $\LPS_{i.j}$).
\begin{observation}\label{obs:cand}
  The $r$-th longest palindrome in a string $T[i.. j]$ is one of the followings:
  \begin{itemize}
    \item[(i)] a maximal palindrome of $T$ whose center is between $c_p+0.5$ and $c_s-0.5$, inclusive,
    \item[(ii)] the shrink of the $q$-th longest palindrome whose center is between $c_p$ and $c_s$, inclusive, for some $q$ with $1 \le q \le r-1$,
    \item[(iii)] the longest palindromic prefix $\LPP_{i,j}$ of $T[i..j]$,
    \item[(iv)]  the longest palindromic suffix $\LPS_{i,j}$ of $T[i..j]$,
    \item[(v)]  a palindromic substring of $T[i.. j]$ whose center is less than $c_p$, which is shorter than $\LPP_{i,j}$, or
    \item[(vi)] a palindromic substring of $T[i.. j]$ whose center is greater than $c_s$, which is shorter than $\LPS_{i,j}$.
  \end{itemize}
\end{observation}
Note that
the candidates (iii) and (iv) are not necessarily maximal palindromes of $T$,
and thus, we cannot merge them with (i) in general.
Also, the candidates (v) and (vi) are derived by Observation~\ref{observation:mirror}.
See also Figure~\ref{fig:classify}.
\begin{figure}[ht]
    \centering
    \includegraphics[width=\linewidth]{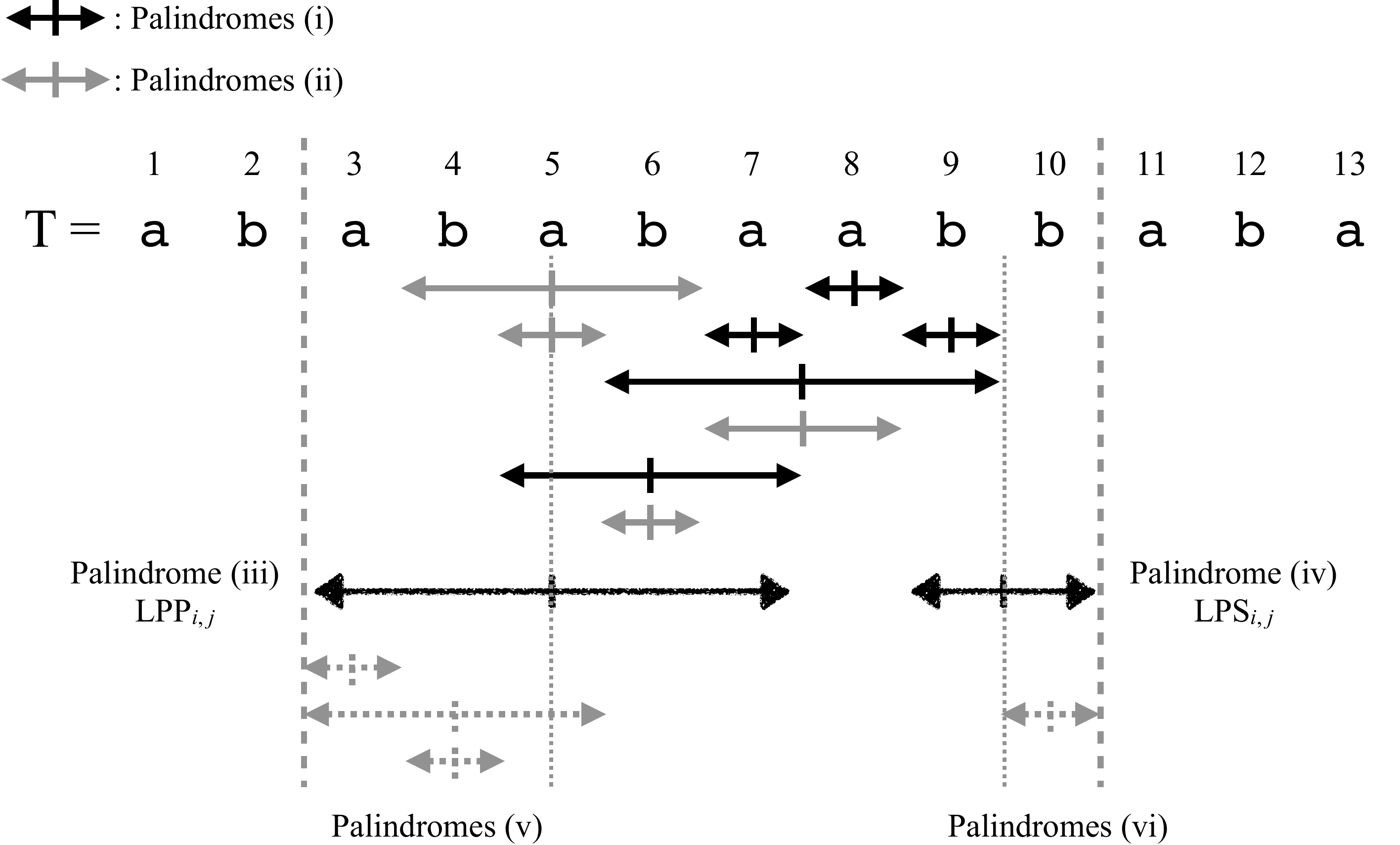}
    \caption{
      Illustration for Observation~\ref{obs:cand}.
      All non-empty palindromic substrings of $T[3.. 10]$ are depicted, and they are categorized into (i)--(vi).
    }
    \label{fig:classify}
\end{figure}

The first and the second candidates (i), (ii) are the same as those of Observation~\ref{obs:topkpal}.
To compute these candidates, we modify the algorithm of Lemma~\ref{lem:topkManacher} and apply it to this problem.
Instead of the sorted array $M[1.. 2n-1]$, we simply use array $\MP[1.. 2n-1]$ of the lengths of maximal palindromes in the positional order.
In the preprocessing, we construct a \emph{top-$k$ Range Maximum Query} (top-$k$ RMQ) data structure on array $\MP$.
The top-$k$ RMQ (a.k.a. sorted range selection query) on an integer array $A[1.. n]$ is, 
given an interval $[i, j] \subseteq [1, n]$ and a positive integer $k \le j-i+1$ as a query,
to output a sorted list of top-$k$ largest elements in subarray $A[i.. j]$.
As for the top-$k$ RMQ, the next result is known:
\begin{theorem}[\cite{brodal2009online}] \label{thm:topkRMQ}
  There is a data structure of size $\order(n)$ which can answer top-$k$ RMQ in $\order(k)$ time for any $k$.
  Also, the data structure can be constructed in $\order(n\log n)$ time.
\end{theorem}

After constructing a {top-$k$ RMQ} data structure on array $\MP$,
we can enumerate top-$k$ longest maximal palindromes of $T$ whose center is between $c_p+0.5$ and $c_s-0.5$ (i.e., candidates (i)).
The second candidate (ii) can be maintained dynamically
by using array $\tilde{R}$ of returned palindromes
which is almost the same as $R$
but contains only palindromes centered between $c_p$ and $c_s$.
The third and fourth candidates (iii) and (iv) are unique.
Thus we can easily treat them.
The fifth and sixth candidates (v), (vi) can be found by palindromic symmetry of $\LPP$ and $\LPS$.
\paragraph*{\bf Algorithm}
Now, we are ready to describe our algorithm.
Given a query interval $[i, j]$ and a query integer $k$,
we first run the algorithm of Theorem~\ref{theorem:conclusion}.
Then, the longest palindromic substring $P_1$ in $T[i.. j]$ and
the longest palindromic prefix/suffix
(equivalently, $c_p$ and $c_s$) of $T[i.. j]$ are obtained.
We set $\tilde{R}[\ell] = \{[s, s+\ell-1]\}$ where $s$ and $\ell$ are the starting position and the length of the longest palindromic substring of $T[i..j]$, respectively.
The second-longest palindrome $P_2$ is the longest one in
$(\{\MP[x]\mid 2c_p < x < 2c_s\} \cup \{T[\alpha+1.. \beta-1]\mid [\alpha, \beta] \in \tilde{R}\}\cup \{\LPP_{i,j}, \LPS_{i,j}\})\setminus \mathcal{R}$
where $\mathcal{R}$ is the set of palindromes that have been returned.
Before detecting $P_2$, $\mathcal{R} = \{[s, s+\ell-1]\}$ holds, and after detecting $P_2$, we update $\mathcal{R} \leftarrow \mathcal{R}\cup\{P_2\}$.
In addition, if $P_2$ is a substring of $\LPP_{i,j}$ and the center of $P_2$ is greater than $c_p$, then
there is the same palindrome as $P_2$ in the opposite position w.r.t. $c_p$.
Thus, we add the (third-longest) palindrome into $\mathcal{R}$ and continue the procedure.
In the case where $P_2$ is a substring of $\LPS_{i,j}$ with a different center from $\LPS_{i,j}$, the same symmetric procedure is applied.
We iterate the above procedure until $k$-th longest palindrome is obtained.
\begin{example}
  We give a running example using Figure~\ref{fig:classify}.
  Assume that we want to find the top-$5$ longest palindromes in substring $T[3.. 10] = \mathtt{ababaabb}$ of string $T = \mathtt{abababaabbaba}$.
  Palindromes in $T[3.. 10]$ are categorized as (i)--(vi).
  Here, the longest maximal palindrome whose center $c$ is $5 < c < 9.5$ is $T[6.. 9]$.
  Also, $\LPP_{3, 10} = T[3.. 7]$ and $\LPS_{3, 10} = T[9.. 10]$ hold.
  Thus, the longest one in $\TopLPal_T([3, 10], 5)$ is $T[3.. 7]$.
  Then $\tilde{R}[5] = \{[3, 7]\}$.
  Thus the second-longest palindrome is either
  $T[6.. 9]$ (a maximal palindrome), $T[9.. 10] = \LPS_{3,10}$, or $T[4..6]$ (, which is the shrink of $T[3.. 7] \in \tilde{R}[5]$).
  Thus we return $T[6.. 9]$ and update $\tilde{R}[4] = \{[6, 9]\}$.
  The second-longest maximal palindrome whose center $c$ is $5 < c < 9.5$ is $T[5.. 7]$.
  Thus the third-longest palindrome is either
  $T[5.. 7]$ (a maximal palindrome), $T[9.. 10] = \LPS_{3,10}$, or $T[4.. 6]$.
  Since there are two longest ones $T[5.. 7]$ and $T[4.. 6]$ with the same length $3$, we return them.
  Also, we remove $[3, 7]$ from $\LPS_{i,j}$ since its shrink $T[4.. 6]$ has been returned.
  Finally, since $T[5.. 7]$ is a palindromic substring of $\LPP_{3,10}$,
  there is a palindrome $T[3..5]$ of length $3$ at the mirror position by Observation~\ref{observation:mirror}.
  So we also return $T[3..5]$ and update $\tilde{R}[3] = \{[5, 7], [4, 6], [3, 5]\}$.
  Then we have returned the top-$5$ palindromes $\mathcal{R} = \{[3,7], [6,9], [5,7], [4,6], [3,5]\}$,
  so the algorithm terminates.
\end{example}

\paragraph*{\bf Analyzing Algorithm}
At each iteration, the longest palindrome in \linebreak
$\{\MP[x]\mid 2c_p < x < 2c_s\} \setminus \mathcal{R}$ can be computed in constant time by answering top-$k$ RMQ in parallel.
The longest palindrome in $\{T[\alpha+1.. \beta-1]\mid [\alpha, \beta] \in \tilde{R}\}\setminus \mathcal{R}$ can be also computed in constant time
since elements in $\tilde{R}$ are sorted by length (cf.~the algorithm of Lemma~\ref{lem:topkManacher}).
Trivially, the longest one in $\{\LPP_{i,j}, \LPS_{i,j}\})\setminus \mathcal{R}$ can be found in constant time.
By exploring the longest elements of the above three sets, the first four candidates in Observation~\ref{obs:topkpal} have been checked.
By the remaining process, the existence of candidates (v) and (vi) in Observation~\ref{obs:topkpal} has also been checked.
Therefore, the proposed algorithm runs correctly, and its time complexity is dominated by the query time for top-$k$ RMQ.
We obtain the next theorem.
\begin{theorem}\label{theorem:conclusion2}
  Given a string $T$ of length $n$ over a linearly sortable alphabet,
  we can construct a data structure of size $\order(n + \pi_s(n))$ in $\order(n + \pi_p(n))$ time
  that can answer any 
  internal top-$k$ longest palindrome query in $\order(k + \pi_q(n, i, j, k))$ time 
  where
  $\pi_p(n)$ is the preprocessing time for top-$k$ RMQ,
  $\pi_s(n)$ is the size of top-$k$ RMQ data structure, and 
  $\pi_q(n, i, j, k)$ is the query time for top-$k$ RMQ.
\end{theorem}
We obtain the next corollary by plugging the result of Theorem~\ref{thm:topkRMQ} into Theorem~\ref{theorem:conclusion2}.
\maincoro*

\section{Conclusions and Open Problems}\label{sec:conc}
In this paper, we considered three variants of the longest palindrome problem
on the input string $T$ of length $n$
and proposed algorithms for them.
The problems are the followings.
\begin{enumerate}
  \item {\bf The internal longest palindrome query},
    which requires returning the longest palindrome appearing in substring $T[i.. j]$.
  \item {\bf The top-$k$ longest palindrome query},
    which requires returning the top-$k$ longest palindrome appearing in $T$.
  \item {\bf The internal top-$k$ longest palindrome query},
    which requires returning the top-$k$ longest palindrome appearing in substring $T[i.. j]$.
\end{enumerate}
Note that every problem is a generalization of the longest palindrome problem,
which can be solved in $\order(n)$ time~\cite{manacher1975new}.
Our proposed data structures are of size $\order(n)$ and can answer every query in optimal time,
i.e., in $\order(1)$ time for the internal longest palindrome query
and in $\order(k)$ time for the top-$k$ queries.
Construction time is $\order(n)$ for the first and the second problem,
and $\order(n \log n)$ time for the third problem.
Note that this $\order(n \log n)$ term is dominated by the preprocessing time for the top-$k$ RMQ~\cite{brodal2009online}.

\paragraph*{Open Problems}
Our results achieved optimal time in terms of order notations.
It will be an interesting open problem to develop a time-space tradeoff algorithm for variants of the longest palindrome problem.
For example, for some parameter $\tau > 1$,
can we design a data structure (except for the input string) of size $\order(n/\tau)$
which can answer the internal longest palindrome query in $\order(\tau)$ time?
One of the other possible directions to reduce space is,
designing a data structure of size $\order(d)$ where $d$ is the number of distinct palindromes occurring in $T$.
It is known that $d$ is at most $n+1$~\cite{droubay2001episturmian} and thus $\order(d) = \order(n)$ in the worst case.
However, in most cases, $d$ is much smaller than $n$.
So it is worthwhile to design such data structures.
For example, the size of the palindromic tree is actually $\order(d)$ rather than $\order(n)$.
Whether the space usage of our data structure can be reduced to $\order(d)$ is open.

\section*{Acknowledgements}
This work was partially supported by JSPS KAKENHI Grant Numbers
JP20H05964~(TH), JP21H05839~(KS), JP22K21273, and JP23H04381~(TM).

\bibliographystyle{plain}
\bibliography{ref}
\end{document}